\documentclass[a4paper,UKenglish,cleveref]{lipics-v2019}

\usepackage{microtype}
\usepackage{amsfonts}
\usepackage{mathtools,textcomp}
\usepackage{booktabs}
\usepackage[vlined]{algorithm2e}
\DontPrintSemicolon

\SetCommentSty{commentFont}

\usepackage{newunicodechar}
\usepackage{silence}
\newunicodechar{✓}{\checkmark}
\newunicodechar{⁻}{^-}
\newunicodechar{⁺}{^+}
\newunicodechar{₋}{_-}
\newunicodechar{₊}{_+}
\newunicodechar{ℓ}{\ell}
\newunicodechar{•}{\item}
\newunicodechar{…}{\dots}
\newunicodechar{≔}{\coloneqq}
\newunicodechar{≤}{\leq}
\newunicodechar{≥}{\geq}
\newunicodechar{≰}{\nleq}
\newunicodechar{≱}{\ngeq}
\newunicodechar{⊕}{\oplus}
\newunicodechar{ϕ}{\varphi}
\newunicodechar{≠}{\neq}
\newunicodechar{¬}{\neg}
\newunicodechar{≡}{\equiv}
\newunicodechar{₀}{_0}
\newunicodechar{₁}{_1}
\newunicodechar{₂}{_2}
\newunicodechar{₃}{_3}
\newunicodechar{₄}{_4}
\newunicodechar{₅}{_5}
\newunicodechar{₆}{_6}
\newunicodechar{₇}{_7}
\newunicodechar{₈}{_8}
\newunicodechar{₉}{_9}
\newunicodechar{⁰}{^0}
\newunicodechar{¹}{^1}
\newunicodechar{²}{^2}
\newunicodechar{³}{^3}
\newunicodechar{⁴}{^4}
\newunicodechar{⁵}{^5}
\newunicodechar{⁶}{^6}
\newunicodechar{⁷}{^7}
\newunicodechar{⁸}{^8}
\newunicodechar{⁹}{^9}
\newunicodechar{∈}{\in}
\newunicodechar{∉}{\notin}
\newunicodechar{⊂}{\subset}
\newunicodechar{⊃}{\supset}
\newunicodechar{⊆}{\subseteq}
\newunicodechar{⊇}{\supseteq}
\newunicodechar{⊄}{\nsubset}
\newunicodechar{⊅}{\nsupset}
\newunicodechar{⊈}{\nsubseteq}
\newunicodechar{⊉}{\nsupseteq}
\newunicodechar{∪}{\cup}
\newunicodechar{∩}{\cap}
\newunicodechar{∀}{\forall}
\newunicodechar{∃}{\exists}
\newunicodechar{∄}{\nexists}
\newunicodechar{∨}{\vee}
\newunicodechar{∧}{\wedge}
\newunicodechar{ℝ}{\mathbb{R}}
\newunicodechar{ℕ}{\mathbb{N}}
\newunicodechar{𝔽}{\mathbb{F}}
\newunicodechar{ℤ}{\mathbb{Z}}
\newunicodechar{·}{\cdot}
\newunicodechar{∘}{\circ}
\newunicodechar{×}{\times}
\newunicodechar{↑}{\uparrow}
\newunicodechar{↓}{\downarrow}
\newunicodechar{→}{\rightarrow}
\newunicodechar{←}{\leftarrow}
\newunicodechar{⇒}{\Rightarrow}
\newunicodechar{⇐}{\Leftarrow}
\newunicodechar{↔}{\leftrightarrow}
\newunicodechar{⇔}{\Leftrightarrow}
\newunicodechar{↦}{\mapsto}
\newunicodechar{∅}{\emptyset}
\newunicodechar{∞}{\infty}
\newunicodechar{≅}{\cong}
\newunicodechar{≈}{\approx}

\newunicodechar{α}{\alpha}
\newunicodechar{β}{\beta}
\newunicodechar{γ}{\gamma}
\newunicodechar{Γ}{\Gamma}
\newunicodechar{δ}{\delta}
\newunicodechar{Δ}{\Delta}
\newunicodechar{ε}{\varepsilon}
\newunicodechar{ζ}{\zeta}
\newunicodechar{η}{\eta}
\newunicodechar{θ}{\theta}
\newunicodechar{Θ}{\Theta}
\newunicodechar{ι}{\iota}
\newunicodechar{κ}{\kappa}
\newunicodechar{λ}{\lambda}
\newunicodechar{Λ}{\Lambda}
\newunicodechar{μ}{\mu}
\newunicodechar{ν}{\nu}
\newunicodechar{ξ}{\xi}
\newunicodechar{Ξ}{\Xi}
\newunicodechar{π}{\pi}
\newunicodechar{Π}{\Pi}
\newunicodechar{ρ}{\rho}
\newunicodechar{σ}{\sigma}
\newunicodechar{Σ}{\Sigma}
\newunicodechar{τ}{\tau}
\newunicodechar{υ}{\upsilon}
\newunicodechar{ϒ}{\Upsilon}
\newunicodechar{ϕ}{\phi}
\newunicodechar{Φ}{\Phi}
\newunicodechar{χ}{\chi}
\newunicodechar{ψ}{\psi}
\newunicodechar{Ψ}{\Psi}
\newunicodechar{ω}{\omega}
\newunicodechar{Ω}{\Omega}
\newunicodechar{𝟙}{\mathds{1}}
\newunicodechar{✗}{\ding{55}}

\usepackage{dsfont}
\newunicodechar{𝟙}{\mathds{1}}
\usepackage{pifont}
\newunicodechar{✗}{\ding{55}}
\usepackage{bm}
\usepackage{xspace}
\usepackage{todonotes}

\usepackage{silence}
\usepackage{arydshln}

\usepackage{tikz}
\usetikzlibrary{calc,arrows.meta}

\makeatletter
\g@addto@macro\@floatboxreset\centering
\makeatother

\newtheorem{fact}[theorem]{Fact}
\crefname{fact}{Fact}{Facts}

\newcommand{\whp}{\textrm{whp}\xspace}

\renewcommand{\O}{\mathcal{O}}
\renewcommand{\exp}[1]{\mathrm{exp}(#1)}
\newcommand{\expb}[1]{\mathrm{exp}\bigg(#1\bigg)}

\DeclareMathOperator{\Po}{Po}
\DeclareMathOperator{\Bin}{Bin}

\let\P=\PP
\def\P{\mathbf{P}}
\def\hP{\mathbf{\hat{P}}}
\def\peel{\mathsf{peel}}
\def\root{\mathrm{root}}
\def\q#1{q^{(#1)}}
\def\sL{\operatorname{\begin{tikzpicture}[scale=0.06,baseline=0.5,line cap=round]
    \draw[very thick] (0,0) -- (-2,2) -- (0,4);
    \draw[very thick] (1.5,0) -- (-0.5,2) -- (1.5,4);
\end{tikzpicture}}}
\def\sR{\operatorname{\begin{tikzpicture}[xscale=-0.06,yscale=0.06,baseline=0.5,line cap=round]
    \draw[very thick] (0,0) -- (-2,2) -- (0,4);
    \draw[very thick] (1.5,0) -- (-0.5,2) -- (1.5,4);
\end{tikzpicture}}}
\def\ero{\mathrm{er}}
\def\co{\mathrm{co}}

\def\hut{\expandafter\hat}

\def\U{\mathcal{U}}

\def\ms{\textrm{\textmu s}}
\def\ns{\textrm{ns}}

\title{Dense Peelable Random Uniform Hypergraphs}
\titlerunning{Dense Peelable Hypergraphs}

\author{Martin Dietzfelbinger}{Technische Universität Ilmenau, Germany}{martin.dietzfelbinger@tu-ilmenau.de}{https://orcid.org/0000-0001-5484-3474}{}
\author{Stefan Walzer}{Technische Universität Ilmenau, Germany}{stefan.walzer@tu-ilmenau.de}{https://orcid.org/0000-0002-6477-0106}{}
\authorrunning{M.\,Dietzfelbinger and S.\,Walzer}
\Copyright{Martin Dietzfelbinger and Stefan Walzer}

\category{}
\relatedversion{}
\supplement{}
\funding{}
\acknowledgements{}

\relatedversion{This paper will be presented at the European Symposium on Algorithms 2019.}
\hideLIPIcs
\nolinenumbers

\begin{document}

\let\paragraph=\subparagraph
    \maketitle
 \begin{abstract}
    We describe a new family of $k$-uniform hypergraphs with independent random edges. The hypergraphs have a high probability of being \emph{peelable}, i.e. to admit no sub-hypergraph of minimum degree $2$, even when the edge density (number of edges over vertices) is close to $1$.
    
    In our construction, the vertex set is partitioned into linearly arranged \emph{segments} and each edge is incident to random vertices of $k$ consecutive segments. Quite surprisingly, the linear geometry allows our graphs to be peeled “from the outside in”. The density thresholds~$f_k$ for peelability of our hypergraphs ($f_3 ≈ 0.918$, $f_4 ≈ 0.977$, $f_5 ≈ 0.992$,~…) are well beyond the corresponding thresholds ($c_3 ≈ 0.818$, $c_4 ≈ 0.772$, $c_5 ≈ 0.702$,~…) of standard $k$-uniform random hypergraphs.
    
    To get a grip on $f_k$, we analyse an idealised peeling process on the random weak limit of our hypergraph family. The process can be described in terms of an operator on $[0,1]^ℤ$ and $f_k$ can be linked to thresholds relating to the operator. These thresholds are then tractable with numerical methods.
    
    Random hypergraphs underlie the construction of various data structures based on hashing, for instance invertible Bloom filters, perfect hash functions, retrieval data structures, error correcting codes and cuckoo hash tables, where inputs are mapped to edges using hash functions.
    The data structures frequently rely on peelability of the hypergraph or peelability allows for simple linear time algorithms. Memory efficiency is closely tied to edge density while worst and average case query times are tied to maximum and average edge size.
    
    To demonstrate the usefulness of our construction, we used our $3$-uniform hypergraphs as a drop-in replacement for the standard $3$-uniform hypergraphs in a retrieval data structure by Botelho et al \cite{BPZ:Practical:2013}. This reduces memory usage from $1.23m$ bits to $1.12m$ bits ($m$ being the input size) with almost no change in running time. Using $k > 3$ attains, at small sacrifices in running time, further improvements to memory usage.
\end{abstract}

\nocite{DR:Towards:2012}


\section{Introduction}
\label{sec:intro}
The \emph{core} of a hypergraph $H= (V,E)$ is the largest sub-hypergraph of $H$ with minimum degree at least $2$. The core can be obtained by \emph{peeling}, which means repeatedly choosing a vertex of degree $0$ or $1$ and removing it (and the incident edge if present) from the hypergraph, until no such vertex exists. If the core of $H$ is empty, then $H$ is called \emph{peelable}.
\subparagraph{The significance of peelability.} Hypergraphs underlie many hashing based data structures and peelability is often necessary for proper operation or allows for simple linear time algorithms. We list a few examples.
\def\fatitem#1{\item {\textbf{#1}}\ }
\begin{itemize}
        \fatitem{Invertible Bloom Lookup Tables.} IBLTs  \cite{GM:Invertable:2011} are based on Bloomier filters \cite{CC:Bloomier_Filters:2008} which are based on Bloom filters \cite{B:Space:1970}. Each element is inserted in several random positions in a hash table. Any cell stores the \textsc{xor} of all elements that have been inserted into it. A \textsc{List-Entries} query on an IBLT can recover all elements of the table precisely if the underlying hypergraph is peelable. Among other things, IBLTs have been used to construct error correcting codes \cite{MV:Biff:2012} and to solve the set reconciliation and straggler identification problems  \cite{EG:StragglerIdentification:2011}.
        \fatitem{Erasure Correcting Codes.} To construct capacity achieving erasure codes, the authors of \cite{LMSS:Efficient_Erasure:2001} consider a hypergraph where $V$ corresponds to parity check bits and $E$ to message bits that were lost during transmission. A message bit is incident to precisely those check bits to which it contributed. Correct decoding hinges on peelability of the hypergraph.

        \fatitem{Cuckoo Hashing and XORSAT.} In the context of cuckoo hash tables \cite{DW07:Balanced:2007,M:Some_Open:2009,PR:Cuckoo:2004} and solving random \textsc{xorsat} formulas \cite{DM:The:3:XORSAT:2002,FP:Sharp:2012,PS:The:Satisfiability:2016}, (partial) peelability of the underlying hypergraph makes placing all (some) keys or solving the linear system (eliminating some variables) particularly simple.
        \fatitem{Retrieval and Perfect Hashing.} The retrieval problem (considered later in \cref{sec:experiments}) occurs in the context of constructing perfect hash functions \cite{BBOVV:Cache-Oblivious-Peeling:14,B:Near-Optimal,BPZ:Simple:2007,BPZ:Practical:2013,MWHC:A_Family:1996}. The known approaches involve finding a solution $z: V → R$ for a system $(\sum_{v ∈ e} z(v) = f(e))_{e ∈ E}$ of equations where $H=(V,E)$ is a hypergraph, $f: E → R$ a function and $R$ a small set. If $R$ is a field, then the incidence matrix of $H$ needs to have full rank over $R$ to guarantee the existence of a solution. If $H$ is peelable however, then the existence of a solution is guaranteed even if $R$ only has a group structure. Moreover, it can be computed in linear time.
\end{itemize}
In these contexts, the hypergraph typically has vertex set $[n] = \{1,…,n\}$ and for each element $x$ of an input set $S$, an edge $e_x ⊂ [n]$ is created with incidences chosen via hash functions. For theoretical considerations, the edges $(e_x)_{x ∈ S}$ are often assumed to be independent random variables. This has proven to be a good model for practical settings, even though perfect independence is not achieved by most practical hash functions. An important choice left to the algorithm designer is the distribution of $e_x$.
\subparagraph{Previous work.} If the distribution is such that $\O(n)$ edges have degree $2$ or less (in particular if $H$ is a graph with $\O(n)$ edges), then – due to the well-known “birthday paradox” – there is a constant probability that an edge is repeated. In that case, $H$ is clearly not peelable. The simplest workable candidate for the distribution of $e_x$ is therefore to pick a constant $k ≥ 3$ and let $e_x$ contain $k$ vertices chosen independently and uniformly at random. We refer to these standard hypergraphs as \emph{$k$-uniform Erdős-Renyi hypergraphs} $H^k_{n,cn}$ where $c$ is the \emph{edge density}, i.e.~the number of edges over the number of vertices. Corresponding \emph{peelability thresholds} $c_k$ have been determined in \cite{Molloy05:Cores-in-random-hypergraphs} meaning if $c < c_k$ then $H^k_{n,cn}$ is peelable with high probability (whp), i.e.~with probability approaching $1$ as $n→∞$ and if $c > c_k$ then $H^k_{n,cn}$ is not peelable whp. The largest threshold is $c_3 ≈ 0.818$. Since the edge density is often tightly linked to a performance metric (e.g. memory efficiency of a dictionary, rate of a code) a density closer to $1$ would be desirable, but we know of only two alternative constructions.

To obtain erasure codes with high rates the authors of \cite{LMSS:Efficient_Erasure:2001}\label{page:LMSS} construct for any $D ∈ ℕ$ hypergraphs with edge sizes in $\{5,…,D+4\}$, average edge size $≈ \ln D+3$ and edge density $1- 1/D$ that are peelable whp. In particular, this yields peelable hypergraphs with edge densities arbitrarily close to $1$. A downside is that the high maximum edge size can lead to worst case query times of $Θ(D)$ in certain contexts. Motivated by this, the author of \cite{R:Mixed:2013} looked into non-uniform hypergraphs with constant maximum edge size. Focusing on hypergraphs with two admissible edge sizes, he found for example that mixing edges of size $3$ and size $21$ yields a family of hypergraphs with peelability threshold $≈0.92$.

\subparagraph{Our construction.} In this paper we introduce and analyse a new distribution on edges that yields $k$-uniform hypergraphs with high peelability thresholds that perform well in practical algorithms.

We call our hypergraphs \emph{fuse graphs} (as in the cord attached to a firecracker). There is an underlying linear geometry and similar to how fire proceeds linearly through a lit fuse,
the peeling process proceeds linearly through our hypergraphs, in the sense that vertices on the inside of the line tend to only become peelable after vertices closer to the end of the line have already been removed. 

Formally, for $k ≥ 3$, $ℓ ∈ ℕ$ and $c ∈ ℝ^+$ we define the family $(F(n,k,c,ℓ))_{n∈ℕ}$ of $k$-uniform fuse graphs as follows. The vertex set is $V = \{1,…,n(ℓ+k-1)\}$ where for $i ∈ I := \{0,…,ℓ+k-2\}$ the vertices $\{in+1,…,(i+1)n\}$ form the $i$-th \emph{segment}\footnote{Denoting the segment size by $n$ instead of the number of vertices is more convenient. Note that $|V| = Θ(n)$ still holds.}. The edge set $E$ has size $cnℓ$. Each edge $e ∈ E$ is independently determined by one uniformly random variable $j ∈ J := \{0,…,ℓ-1\}$ denoting the \emph{type} of $e$ and $k$ independent random variables $o₀,…,o_{k-1}$ uniformly distributed in $[n]$, yielding $e = \{(j + t)n + o_t \mid t ∈ \{0,…,k-1\} \}$. In other words, $e$ contains one uniformly random vertex from each segment $j,j+1,…,j+k-1$. There may be repeating edges but the probability that his happens is $\O(1/n)$. The edge density $c \frac{ℓ}{ℓ+k-1}$ approaches $c$ for $ℓ \gg k$.

\subparagraph{Results.} Let the \emph{peelability threshold} for $k$-ary fuse graphs be defined as
\[ f_k := \sup\{ c ∈ ℝ^+ \mid ∀ℓ ∈ ℕ: \Pr[F(n,k,c,ℓ) \text { is peelable}] \stackrel{n→∞}\longrightarrow 1 \}. \]
Our Main Theorem relates $f_k$ to the \emph{orientability threshold} $c_k^*$ of $k$-ary Erdős-Renyi hypergraphs and the \emph{erosion threshold} $\ero_k$ defined in the technical part of our paper.

\begin{theorem}
    \label{thm:main} For any $k ≥ 3$ we have $\ero_k ≤ f_k ≤ c_k^*$.
\end{theorem}
\noindent The orientability thresholds $c_k^*$ are known exactly \cite{DGMMPR:Tight:2010,FP:Sharp:2012,FM:Maximum:2012} and we determine lower bounds on the erosion thresholds $\ero_k$. As shown in \cref{tab:upper-and-lower-bounds}, this makes it possible to narrow down $f_k$ to an interval of size $10^{-5}$ for all $k ∈ \{3,…,7\}$. 

\begin{table}
    \begin{tabular}{cccccccc}
        \toprule
        $k$ & 3 & 4 & 5 & 6 & 7\\
        \midrule
        $b_k$&0.9179352469&0.9767692112&0.9924345766&0.9973757381&0.9990561294\\
        $c_k^*$&0.9179352767&0.9767701649&0.9924383913&0.9973795528&0.9990637588\\
        $B_k$&0.9179353065&0.9767711186&0.9924422067&0.9973833675&0.9990713882\\
        \midrule
        $⇒ f_k ≈$ &0.917935&0.97677&0.99243&0.99738&0.99906\\
        \bottomrule
    \end{tabular}
    \caption{The erosion thresholds $\ero_k$ and peelability thresholds $f_k$ for $k$-ary fuse graphs satisfy $b_k ≤ \ero_k ≤ f_k ≤ c_k^*$. The values $B_k$ play a role in \cref{sec:approximating}.}
    \label{tab:upper-and-lower-bounds}
\end{table}

\subparagraph{Outline.}
The paper is organised as follows. In \cref{sec:peeling-operator} we idealise the peeling process by switching to the \emph{random weak limit} of our hypergraphs, and capture the essential behaviour of the process in terms of an operator $\hP$ acting on functions $q: ℤ → [0,1]$. For this operator, we identify the properties of being \emph{eroding} and \emph{consolidating} as well as corresponding thresholds $\ero_k$ and $\co_k$ in \cref{sec:eroding-consolidating}. We then prove the “$\ero_k ≤ f_k$” part of our theorem in \cref{sec:wrapping-up} and give numerical approximations of $\ero_k$ and $\co_k$ in \cref{sec:approximating}. The comparatively simple “$f_k ≤ c_k^*$” part of our theorem is independent of these considerations and is proved in \cref{sec:upperbound}. Finally, in \cref{sec:experiments} we demonstrate how using our hypergraphs can improve the performance of practical retrieval data structures.






\section{The Peeling Process and Idealised Peeling Operators}
\label{sec:peeling-operator}

In this section we consider how the probabilities for vertices to “survive” $r ∈ ℕ$ rounds of peeling changes from one round to the next. In the classical setting this could be described by a function, mapping the old probability to the new one \cite{Molloy05:Cores-in-random-hypergraphs}. In our case, however, there are distinct probabilities for each segment of the graph. Thus we need a corresponding operator $\hP$ that acts on \emph{sequences} of probabilities. Conveniently, it will be independent of $n$ and $ℓ$.

We almost always suppress $n,k,c,ℓ$ in notation outside of definitions, assuming $n$ to be large. Big-$\O$ notation refers to $n → ∞$ while $k,c,ℓ$ are constant.

Consider the parallel peeling process $\peel(F)$ on $F = F(n,k,c,ℓ)$. In each \emph{round} of $\peel(F)$, all vertices of degree $0$ or $1$ are determined and then deleted simultaneously. Deleting a vertex implicitly deletes incident edges. We also define the \emph{rooted peeling process} $\peel_v(F)$ for any vertex $v ∈ V$, which behaves exactly like $\peel(F)$ except that the special vertex $v$ may only be deleted if it has degree $0$, not if it has degree $1$. For any $i ∈ I$ and $r ∈ ℕ₀$ we let $\q{r}(i) = \q{r}(i,n,k,c,ℓ)$ be the probability that a vertex $v$ of segment $i$ survives $r$ rounds of $\peel_v(F)$, i.e.~is not deleted. Note that the probability is well-defined as vertices of the same segment are symmetric.

 
By definition, $\q{0}(i) = 1$ for all $i ∈ I$. Whether a vertex $v$ of segment $i ∈ I$ survives $r > 0$ rounds is a function of its $r$-neighbourhood $N(n,v,r)$, i.e.~the set of vertices and edges of $F$ that can be reached from $v$ by traversing at most $r$ hyperedges.
 
It is standard to consider the \emph{random weak limit} of $F$ to get a grip on the distribution of $N(n,v,r)$ and thus on $\q{r}(i)$. Intuitively, we identify a (possibly infinite) random tree that captures the local characteristics of $F$ for $n → ∞$. See \cite{AS:Objective_Method:2004} for a good survey with examples and details on how to formally define the underlying topology and metric space.
In the limit, the binomially distributed vertex degrees (e.g. $\Bin(cnℓ,\frac{1}{nℓ})$ for vertices of segment $0$) become Poisson distributed ($\Po(c)$ for segment $0$). Short cycles are not only rare but non-existent and certain weakly correlated random variables become perfectly independent.
 

 \begin{definition}[Limiting Tree]
    \label{def:T}
    Let $k, ℓ ∈ ℕ$, $c ∈ ℝ^+$ and $i ∈ I$. The random (possibly infinite) hypertree $T_i = T_i(k,c,ℓ)$ is distributed as follows.
    
    $T_i$ has a root vertex $\root(T_i)$ of segment\footnote{In the current context, the segment of a vertex is an abstract label. There can be an unbounded number of vertices of each segment.} $i$ which, for each $j ∈ \{i-k+1,…,i\} ∩ J$, has $d_j \sim \Po(c)$ \emph{child edges} of type $j$. Each child edge of type $j$ is incident to $k-1$ (fresh) \emph{child vertices} of its own, one for each segment $i' ∈ \{j,…,j+k-1\} \setminus \{i\}$. The sub-hypertree at such a child vertex of segment $i'$ is distributed recursively (and independently of its sibling-subtrees) according to $T_{i'}$.
 \end{definition}
\noindent Since all arguments are standard in contexts where local weak convergence plays a role, we state the following lemma without proof. For instance, a full argument to show a similar convergence is given in \cite{Leconte:Cuckoo:2013}. See also \cite{K:Poisson:2006} for the related technique of Poissonisation.
 \begin{lemma}
    \label{lem:localConvergenceOfH}
    Let $r ∈ ℕ$ be constant. Let further $N(n,v,r)$ be the $r$-neighbourhood of a vertex $v$ of segment $i$ in $F$ and $T_i^{(r)}$ the $r$-neighbourhood of $\root(T_i)$, both viewed as undirected and unlabelled hypergraphs. Then $N(n,v,r)$ converges in distribution to $T_i^{(r)}$ as $n → ∞$.
 \end{lemma}
 \noindent We now direct our attention to survival probabilities in the idealised peeling processes $(\peel_{\root(T_i)}(T_i))_{i ∈ I}$, which are easier to analyse than those of $\peel_v(F)$.
 
 \begin{lemma}
    \label{lem:peelingT}
    Let $r ∈ ℕ₀$ be constant and $\q{r}_T(i) = \q{r}_T(i,k,c,ℓ)$ be the probability that $\root(T_i)$ survives $r$ rounds of $\peel_{\root(T_i)}(T_i)$ for $i ∈ I$. Then
    \[ \q{r+1}_T(i) = 1 - \expb{-c\sum_{j ∈ \{i-k+1,…,i\}∩J}\ \ \prod_{\substack{j ≤ i' < j+k\\i' ≠ i}} \q{r}_T(i')} \qquad\text{for } i ∈ I.\]
 \end{lemma}
 
 \begin{proof}
    Let $i ∈ I$ and $v = \root(T_i)$. Assume $j ∈ \{i-k+1,…,i\} ∩ J$ is the type of some edge $e$ incident to $v$. Edge $e$ survives $r$ rounds of $\peel_v(T_i)$ if and only if all of its incident vertices survive these $r$ rounds. Since $v$ itself may not be deleted by $\peel_v(T_i)$ as long as $e$ exists, the relevant vertices are the $k-1$ child vertices, one for each segment $i' ∈ \{j,…,j+k-1\}-\{i\}$. Call these $w₁,…,w_{k-1}$ and denote the subtrees rooted at those vertices by $W₁,…,W_{k-1}$. Now consider the peeling processes $\peel_{w₁}(W₁), …, \peel_{w_{k-1}}(W_{k-1})$. Assume one of them, say $\peel_{w_s}(W_s)$, deletes $w_s$ in round $r' ≤ r$, meaning $w_s$ has degree $0$ before round $r'$. It follows that $w_s$ has degree at most $1$ before round $r'$ in $\peel_v(T_i)$, meaning $\peel_v(T_i)$ deletes $e$ in round $r'$ (or earlier). Conversely, if none of $\peel_{w₁}(W₁),…,\peel_{w_{k-1}}(W_{k-1})$ delete their root vertex within $r$ rounds, then $w₁,…,w_{k-1}$ have degree at least $2$ after round $r$ of $\peel_{v}(T_i)$ and $e$ survives round $r$ of $\peel_v(T_i)$. This makes the probability for $e$ to survive $r$ rounds of $\peel_v(T_i)$ equal to $p_{ij} := \prod_{j ≤ i' < j+k, i' ≠ i} \q{r}_T(i')$. Since the number $m_{ij}$ of edges of type $j$ incident to $v$ has distribution $m_{ij} \sim \Po(c)$, the number $m_{ij}'$ of edges of type $j$ incident to $v$ surviving $r$ rounds of $\peel_v(T_i)$ is a correspondingly \emph{thinned out} variable, namely $m_{ij}' \sim \Bin(m_{ij}, p_{ij})$, which means $m_{ij}' \sim \Po(cp_{ij})$.
    
    The claim now follows by observing that $v$ survives $r+1$ rounds of $\peel_v(T_i)$ if and only if at least one of its child edges survives $r$ rounds of $\peel_v(T_i)$:
    \begin{align*}
        \q{r+1}_T(i) &= \Pr[v \text{ survives $r+1$ rounds of $\peel_v(T_i)$}]
        = 1 - \Pr\big[\bigcap_{j ∈ \{i-k+1,…,i\}∩J}\! \{ m_{ij}' = 0\}\big]\\
        &= 1 - \ \ \ \prod_{\mathclap{j ∈ \{i-k+1,…,i\}∩J}}\ \ \Pr[m_{ij}' = 0]
        = 1 - \ \ \ \prod_{\mathclap{j ∈ \{i-k+1,…,i\}∩J}}\ \ \  \exp{-cp_{ij}}
        = 1 - \exp{-c\ \ \ \sum_{\mathclap{j ∈ \{i-k+1,…,i\}∩J}}\ \ \  p_{ij}}.
    \end{align*}
    Replacing $p_{ij}$ with its definition completes the proof.
 \end{proof}

For convenience we define, for $k ≥ 3, ℓ ∈ ℕ$ and $c ∈ ℝ^+$, the operator ${\P} ={\P}(k,c,ℓ)$, which maps any $q : I → [0,1]$ to $\P q : I → [0,1]$ with
 \[ (\P q)(i) = 1 - \expb{-c \sum_{j ∈ \{i-k+1,…,i\}∩J}\ \ \prod_{\substack{j ≤ i' < j+k\\i' ≠ i}} q(i')} \quad \text{for } i ∈ I.\]
 Together Lemmas \ref{lem:localConvergenceOfH} and \ref{lem:peelingT} imply that $\P$ can be used to approximate survival probabilities.
 \def\err{\mathrm{err}}
 \begin{corollary}
    \label{cor:approxOfq}
    Let $r ∈ ℕ₀$ be constant. Then for all $i ∈ I$
    \[ \P^r \q{0}(i) \stackrel{\textup{def}}= \P^r \q{0}_T(i) \stackrel{\textup{Lem\,\ref{lem:peelingT}}}= \q{r}_T(i) \stackrel{\textup{Lem\,\ref{lem:localConvergenceOfH}}}= \q{r}(i)\pm o(1). \]
 \end{corollary}
\noindent To obtain \emph{upper} bounds on survival probabilities, we may remove the awkward restriction “$∩\,J$” in the definition of $\P$. We define $\hP = \hP(k,c)$ as mapping $q : ℤ → [0,1]$ to $\hP q : ℤ → [0,1]$ with
\[ (\hP q)(i) = 1 - \expb{-c\sum_{j = i-k+1}^i\ \ \prod_{\substack{j ≤ i' < j+k\\i' ≠ i}} q(i')} \quad \text{for } i ∈ ℤ.\]
Note that $\hP$ does not depend on $ℓ$ or $n$. To simplify notation, we assume that the old operator $\P$ also acts on functions $q : ℤ → [0,1]$, ignoring $q(i)$ for $i ∉ I$, and producing $\P q : ℤ → [0,1]$ with $\P q(i) = 0$ for $i ∉ I$. We also extend $\q{0}$ to be $𝟙_{I} : ℤ → [0,1]$, i.e. the characteristic function on $I$, essentially introducing vertices of segments $i ∉ I$ which are, however, already deleted with probability $1$ before the first round begins.
Note that while $\q{r}(i)$ and $\q{r}_T(i)$ are by definition non-increasing in $r$, this is not the case for $(\hP^r \q{0})(i)$. For instance, $\hP^{r} \q{0}$ has support $\{0-r,…,ℓ+k-2+r\}$, which grows with $r$.\footnote{It is still possible to interpret $\hP^r \q{0}(i)$ as survival probabilities in more symmetric extended versions $\hat{T}_i$ of the tree $T_i$, but we will not pursue this.}
The following lemma lists a few easily verified properties of $\hP$. All inequalities between functions should be interpreted point-wise.
\begin{lemma}
    \label{lem:properties-hP}
    \begin{enumerate}[(i)]
            • $∀q: ℤ → [0,1]\colon \P q ≤ \hP q$.
            • $\hP$ commutes with the shift operators $\sL$ and $\sR$ defined via $(\sL q)(i) = q(i+1)$ and $(\sR q)(i) = q(i-1)$. In other words, we have $∀q: ℤ → [0,1] \colon \hP (\sL q) = \sL (\hP q)  ∧ \hP (\sR q) = \sR (\hP q)$.
            
            • $\hP$ is monotonic, i.e.~$∀q,q': ℤ → [0,1] \colon q ≤ q' ⇒ \hP q ≤ \hP q'$.
            • $\hP$ respects monotonicity, i.e.~if $q(i)$ is (strictly) increasing in $i$, then so is $(\hP q)(i)$.
    \end{enumerate}
\end{lemma}

\section{Two Fixed Points Battling for Territory}
\label{sec:eroding-consolidating}

\def\const{\mathrm{const}}
\def\conv#1{\stackrel{#1}{\longrightarrow}}

In this section we define the \emph{erosion} and \emph{consolidation thresholds} at which the behaviour of $\hP$ changes in crucial ways.

First, we require a few facts about the function $f: [0,1] → [0,1]$ mapping $x ↦ 1-e^{-ckx^{k-1}}$. It appears in the analysis of cores in $k$-ary Erdős-Renyi hypergraphs $H_{n,cn}^k$, essentially mapping the probability $ρ_r$ for a vertex to survive $r$ rounds of peeling to the probability $ρ_{r+1} = f(ρ_r)$ to survive $r+1$ rounds of peeling, see \cite[page 5]{Molloy05:Cores-in-random-hypergraphs}\footnote{Our setting corresponds to the choices $(r_{\textrm{Molloy}},k_{\textrm{Molloy}},c_{\textrm{Molloy}}) = (k,2,c·(k-1)!)$}.

The threshold $c_k$ for the appearance of a core in $H_{n,cn}^k$ turns out to be the threshold for the appearance of a non-zero fixed point of $f$. The following is implicit in the analysis.
%
\begin{fact}[{\cite[Proofs of Lemmas 3 and 4]{Molloy05:Cores-in-random-hypergraphs}}]\ 
    \begin{enumerate}[(i)]
        • For $c < c_k$, \,$f$ has only the fixed point $f(0) = 0$, with $f'(0) < 1$.
        • For $c > c_k$, there are exactly three fixed points $0$, $ξ₁ = ξ₁(k,c)$ and $ξ₂ = ξ₂(k,c)$ where $f'(ξ₁) > 1$ while $f'(0),f'(ξ₂) < 1$.
    \end{enumerate}
\end{fact}
This implies the following behaviour of applying $f$ repeatedly to a starting value $x$. This should be immediately clear from the sketches we give on the right.

\parbox{0.96\textwidth}{\small\begin{equation}
    \ \hspace{-1cm} f^r(x) \conv{r → ∞} \begin{cases}
        0 &\hspace{-0.3cm}\text{ if $c < c_k ∧ x ∈ [0,1]$,}\\
        0 &\hspace{-0.3cm} \text{ if $c > c_k ∧ x ∈ [0,ξ₁)$,}\\
        ξ₂ &\hspace{-0.3cm}\text{ if $c > c_k ∧ x ∈ (ξ₁,1]$.}
    \end{cases}\raisebox{-1.7cm}{\includegraphics[page=1,scale=0.95]{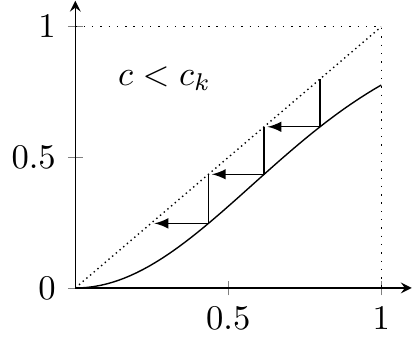}
    \includegraphics[page=2,scale=0.95]{figure-f-iteration/figure.pdf}}
    \label{eq:iterating-f}
\end{equation}}
%
%

Note that $f$ captures the behaviour of $\hP$ on constant functions $\const_x(i) := x$, in the sense that $\hP \const_x = \const_{f(x)}$. For $c < c_k$ we therefore have for all $i ∈ I$
\[  \P^r \q{0}(i) \stackrel{\text{Cor \ref{cor:approxOfq}}}= \q{r}(i)\pm o(1) \text{ and } \P^r \q{0} ≤ \hP^r \q{0} ≤ \hP^r \const_1 = \const_{f^r(1)} \conv{r → ∞} \const_0. \]
In conjunction with a later lemma, this is sufficient to show that $F$ is peelable \whp in this case. A similar argument for $c = c_k$ is possible as well. Our focus from now on is therefore on the interesting case $c > c_k$ where the three distinct fixed points $0$, $ξ₁$, $ξ₂$ of $f$ exist.

We give an intuitive account of the phenomenon underlying the following steps before continuing formally. Due to \eqref{eq:iterating-f} we have
\[ \hP^r \const_x \conv{r → ∞} \begin{cases}
    \const₀ \quad \text{for $x < ξ₁$}\\
    \const_{ξ₂} \quad \text{for $x > ξ₁$.}\\
\end{cases}\]
\def\step{\mathrm{step}}
\def\ibord{i_{\mathrm{bord}}}
Now consider what happens if we iterate $\hP$ on a function that is “torn” between these two cases. Concretely, let us consider the function $\step₀¹$ where we define $\step_x^y : ℤ → [0,1]$ to have value $y$ on $ℕ₀$ and value $x$ on negative inputs.
Should we expect $\hP^r \step₀¹$ to converge to $\const₀$ or $\const_{ξ₂}$ as $r$ increases? It turns out both is possible, depending on $c$.

Speaking more generally, let $q: ℤ → [0,1]$ be any function. If $N(i) := \{i-k+1,…,i+k-1\} \setminus \{i\}$ then $\hP q(i)$ depends (monotonically) on $(q(i'))_{i' ∈ N(i)}$. It is clear that if $q(i') < ξ₁$ for all $i' ∈ N(i)$, then $\hP q(i) < ξ₁$ as well. Similarly, if $q(i') > ξ₁$ for all $i' ∈ N(i)$ then $\hP q(i) > ξ₁$. If, however, there are indices $i'₁,i'₂ ∈ N(i)$ with $q(i'₁) < ξ₁ < q(i'₂)$ then $\hP q(i)$ could be above or below $ξ₁$; in this case we call the index $i$ \emph{contested}. 

The contested area of $\step₀¹$ is $[-k+1,k-2]$. Iterating $\hP$ we obtain $\hP^r \step₀¹$ for $r ∈ ℕ₀$. For all $r ∈ ℕ₀$ the contested area is an interval of size $2k-2$ with all values to the left of it (towards $-∞$) less than $ξ₁$ and all values to the right of it (towards $∞$) bigger than $ξ₁$. However, the contested area may \emph{shift}. If the domain of values bigger than $ξ₁$ is shrinking (\emph{“eroding”}), then we see convergence to $\const₀$. If conversely it is growing (“consolidating”), then we see convergence to $\const_{ξ₂}$. In \cref{fig:movingFronts} we visualise these effects. There is only a small range of values $c$ where both fixed points seem equally “strong” and the same area remains perpetually contested.



\begin{figure}[htbp]
    {\includegraphics[page=1,scale=0.9]{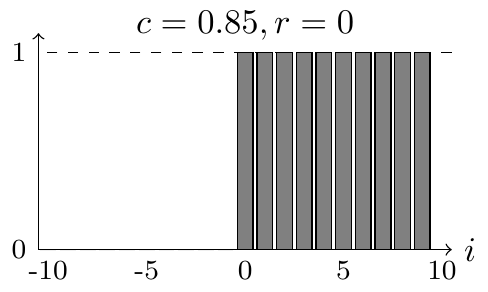}~
    \includegraphics[page=2,scale=0.9]{figure-for-intuition/figure.pdf}~
    \includegraphics[page=3,scale=0.9]{figure-for-intuition/figure.pdf}}\\
    {\includegraphics[page=4,scale=0.9]{figure-for-intuition/figure.pdf}~
    \includegraphics[page=5,scale=0.9]{figure-for-intuition/figure.pdf}~
    \includegraphics[page=6,scale=0.9]{figure-for-intuition/figure.pdf}}
    
    
    \caption{Depiction of $\hP^r \step₀¹$ for $c ∈ \{0.85,0.95\}$ and $r ∈ \{0,5,25\}$ on the range $i ∈ \{-10,…,10\}$. The phenomenon of \emph{erosion} can be seen on the top with the plot seemingly moving towards the right between $r = 5$ and $r = 25$. Similarly, \emph{consolidation} can be seen on the bottom.}
    \label{fig:movingFronts}
\end{figure}

With this in mind, we make the following definitions. For a compact formulation in the coarse terms of shifts (“$\sL$”, “$\sR$”) and point-wise inequalities (“$<$”, “$>$”) we use slightly different step functions.

\begin{definition}
    Let $k ≥ 3$, $c ∈ ℝ^+$ and $\hP = \hP(k,c)$ as above. We say
    \begin{align*}
        \text{$\hP$ is eroding if } &∃ R ∈ ℕ\colon \hP^R \step_{ξ₁/2}^1 < \sR \step_{ξ₁/2}^1\\
        \text{and $\hP$ is consolidating if } & ∃ R ∈ ℕ\colon \hP^R \step_{0}^{(ξ₁+ξ₂)/2} > \sL \step_{0}^{(ξ₁+ξ₂)/2}.
    \end{align*}
\end{definition}
\noindent We define the corresponding \emph{erosion} and \emph{consolidation thresholds} as
\[
    \ero_k = \sup\{c ∈ ℝ^+ \mid \hP(k,c) \text{ is eroding} \},\qquad
    \co_k = \inf\{c ∈ ℝ^+ \mid \hP(k,c) \text{ is consolidating} \}.
\]
Note that $c < \ero_k$ implies that $\hP(k,c)$ is eroding and $c > \co_k$ implies $\hP(k,c)$ is consolidating as would be expected. This uses that the definition of $\hP$ is monotonic in $c$.

The following lemma states that erosion (consolidation) are sufficient conditions for $\const_{0}$ ($\const_{ξ₂}$) to “win the battle” when iterating $\hP$ on $\step₀¹$.
\begin{lemma}
    \label{lem:iterated-eroding}
    Let $k ≥ 3$.
    \begin{enumerate}[(i)]
            • \label{lemitem:iterated-eroding} If $c < \ero_k$ and $i ∈ ℤ$, then $\hP^r \step₀¹(i) \conv{r → ∞} 0$.
            • \label{lemitem:iterated-consolidating} If $c > \co_k$ and $i ∈ ℤ$, then $\hP^r \step₀¹(i) \conv{r → ∞} ξ₂$.
            • \label{lemitem:not-left-and-eroding} $\ero_k ≤ \co_k$.
    \end{enumerate}
\end{lemma}

\begin{proof}
    \begin{enumerate}[(i)]
            • Let $R ∈ ℕ$ be the witness to the fact that $\hP(k,c)$ is eroding and $i ∈ ℤ$ arbitrary.
            \begin{align*}
                \lim_{r → ∞} (\hP^r\step₀¹)(i) &≤ \lim_{r → ∞} (\hP^r\step_{ξ₁/2}¹)(i) = \lim_{r → ∞} (\hP^r ((\hP^{R})^{kr}\step_{ξ₁/2}¹))(i)\\
                &≤ \lim_{r → ∞} (\hP^r (\sR^{kr} \step_{ξ₁/2}¹))(i) = \lim_{r → ∞} (\sR^{kr}(\hP^r \step_{ξ₁/2}¹))(i)\\
                &= \lim_{r → ∞} (\hP^r \step_{ξ₁/2}¹)(i-kr) = \lim_{r → ∞} (\hP^r \const_{ξ₁/2})(i-kr)\\
                &= \lim_{r → ∞} \const_{f^{r}(ξ₁/2)}(i-kr) = \lim_{r → ∞} f^{r}(ξ₁/2) = 0.
            \end{align*}
            When replacing $\step_{ξ₁/2}¹$ by $\const_{ξ₁/2}$ we exploited that $(\hP q)(i)$ depends only on the values $q(i')$ for $i' ∈ \{i-k+1,…,i+k-1\}$ and thus $(\hP^r q)(i)$ depends only on the values $q(i')$ for $i' ∈ [i-(k-1)r,i+(k-1)r]$.
            • The proof is analogous to the proof of (\ref{lemitem:iterated-eroding}).
            • This is clear, since the implications of (\ref{lemitem:iterated-eroding}) and (\ref{lemitem:iterated-consolidating}) are mutually exclusive.
            \qedhere
    \end{enumerate}
\end{proof}

\section{Erosion is Sufficient for Peeling}
\label{sec:wrapping-up}

We now connect the phenomenon of erosion to the survival probabilities $\q{R}(i)$ we were originally interested in. For $c < \ero_k$ and any $ℓ ∈ ℕ$,  they can be made smaller than any $δ > 0$ in $R = R(δ,ℓ)$ rounds. For $c > \co_k$ and $ℓ$ sufficiently large, no constant number of rounds suffices to reduce all survival probabilities below $ξ₁$.

\begin{lemma} Let $k ≥ 3$.
    \label{lem:back-to-peeling}
    \begin{enumerate}[(i)]
        • \label{lemitem:almost-peelable} If $c < \ero_k$ then $∀ℓ ∈ ℕ,δ > 0\colon ∃R,N ∈ ℕ\colon ∀n ≥ N, i ∈ I: \q{R}(i) < δ$.
        • \label{lemitem:not-peelable-in-constant-rounds} If $c > \co_k$ then $∃L = L(k,c):∀ℓ ≥ L\colon ∃i ∈ I: \lim\limits_{r → ∞}\lim\limits_{n → ∞} \q{r}(i) > ξ₁$.
    \end{enumerate}
\end{lemma}

\begin{proof}
    \begin{enumerate}[(i)]
        • Let $ℓ ∈ ℕ$ and $δ > 0$ be arbitrary constants. Using (\ref{lemitem:iterated-eroding}) from Lemma \ref{lem:iterated-eroding}, there exists a constant $R$ such that $\hP^R \step₀¹(i) ≤ δ/2$ for all $i ∈ I$. Therefore for $i ∈ I$:
        \begin{align*}
            \q{R}(i) &\stackrel{\text{Cor \ref{cor:approxOfq}}}= (\P^R \q{0}) (i) + o(1) ≤ (\hP^R \q{0}) (i) + o(1) ≤ (\hP^R \step₀¹)(i) ≤ δ/2 + o(1).
        \end{align*}
        which implies the existence of an appropriate $N ∈ ℕ$.
        • Let $R ∈ ℕ$ be the witness to the fact that $\hP(k,c)$ is consolidating and let $ℓ≥L(k,c) := 4d$ for $d = (k-1)R$. Consider the function $q^* : ℤ → [0,1]$ defined as $q^* = 𝟙_{\{d,…,ℓ-d-1\}} · (ξ₁+ξ₂)/2$, i.e. the function with value $(ξ₁+ξ₂)/2$ on its support $\{d,…,ℓ-d-1\}$. For any $d ≤ i < ℓ-2d$ we have
        \[ \P^R q^*(i) = \hP^R q^*(i) = \hP^R \step_{0}^{(ξ₁+ξ₂)/2}(i-d) ≥ \sL \step_{0}^{(ξ₁+ξ₂)/2}(i-d) = (ξ₁+ξ₂)/2 = q^*(i).\]
        For the first equality, we exploited that $i$ is so far from the borders of $I = \{0,…,ℓ-1\}$ that there is no difference between $\P$ and $\hP$. For the second equality we used that only the values of $q^*$ on $\{i-d,…,i+d\}$ play a role and $q^*$ is a (shifted) step function on that domain. By mirroring, the same argument can be made to get $\P^R q^*(i) ≥ q^*(i)$ for $2d ≤ i < ℓ - d$ as well and thus the point-wise inequality $\P^R q^* ≥ q^*$. Since $\q{0} ≥ q^*$ we get
        \[ \lim_{r → ∞} \lim_{n→∞} \q{r} \stackrel{\text{Cor \ref{cor:approxOfq}}}= \lim_{r → ∞} \P^r 𝟙_I ≥ \lim_{r → ∞} \P^r q^* ≥ q^*.\]
        Since $q^*$ exceeds $ξ₁$ on $\{d,…,ℓ-d-1\}$, this implies the claim.\qedhere
    \end{enumerate}
\end{proof}
While \cref{lem:back-to-peeling}(\ref{lemitem:almost-peelable}) is sufficient to show that that all but a $δ$-fraction of the vertices is peeled whp if $c < \ero_k$, we still need the following combinatorial argument that shows that whp no non-empty core is contained within the remaining vertices. Arguments such as these are standard, many similar ones can be found for instance in \cite{FKP:The_Multiple:2011,FP:Sharp:2012,Luczak:A-simple-solution,L:A_New_Approach:2012,Luczak:Size-and-connectivity-of-the-k-core:1991,Molloy05:Cores-in-random-hypergraphs,MPW:DoubleHashing:2018}.

\begin{lemma}
    \label{lem:no-tiny-core}
    For any $k ≥ 3$, $ℓ ∈ ℕ$ and $c ∈ (0,1)$ there exists $δ = δ(k,ℓ) > 0$ such that the following holds whp. For any non-empty set $V' ⊆ V$ of at most $δ|V|$ vertices of $F = (V,E)$, there exists $v ∈ V'$ of degree at most $1$ in the sub-hypergraph of $H$ induced by $V'$.
\end{lemma}
\begin{proof}
    In the course of the proof we will implicitely encounter positive upper bounds on $δ$ in terms of $k$ and $ℓ$. Any $δ > 0$ small enough to respect these bounds is suitable. We consider the events $(W_{s,t})_{\smash{k ≤ s ≤ δ|V|, \frac{2s}{k} ≤ t ≤ |E|}}$ that some small set $V'$ of size $s$ induces $t$ edges. If \emph{none} of the events occur, then all such $V'$ induce less than $2|V'|/k$ edges and therefore induce hypergraphs with average degree less than $2$, so a vertex of degree at most $1$ exists in each of them.
    
    It is thus sufficient to show that $\Pr[\bigcup_s \bigcup_t W_{s,t}] = \O(1/n)$. We shall use a first moment argument.
    First note that $F$ has duplicate edges with probability $\binom{cnℓ}{2}(ℓn^k)^{-1} = \O(n^{-1})$, so we restrict our attention to $F$ without duplicate edges.
    Given $s$ and $t$ there are $\binom{(ℓ+k-1)n}{s}$ ways to choose $V'$ and at most $\binom{s^k}{t}$ ways to choose which $k$-tuples of vertices in $V'$ induce an edge. The probability that any given $k$-tuple actually does induce an edge is either zero if the $k$ vertices are not of consecutive segments or $1 - (1- (ℓn^k)^{-1})^{cnℓ} ≤ \frac{cn}{n^k} = \frac{1}{n^{k-1}}$. Thus, using constants $C, C', C'',C''' ∈ ℝ^+$ (that may depend on $k$ and $ℓ$) where precise values do not matter, we get
    \begin{align*}
        \ \hspace{-0.8cm}\Pr[\bigcup_{s = k}^{δ|V|} \bigcup_{t = \frac{2s}{k}}^{|E|} &W_{s,t}] ≤ \sum_{s = k}^{δ|V|}\sum_{t = \frac{2s}{k}}^{|E|}\Pr[W_{s,t}] ≤ \sum_{s = k}^{δ|V|}\sum_{t = \frac{2s}{k}}^{|E|} \binom{(ℓ+k-1)n}{s} \binom{s^k}{t} \bigg(\frac{1}{n^{k-1}}\bigg)^t\\
        &≤ \sum_{s = k}^{δ|V|}\sum_{t = \frac{2s}{k}}^{|E|} \bigg(\frac{e(ℓ+k-1)n}{s}\bigg)^{s} \bigg(\frac{es^k}{tn^{k-1}}\bigg)^{t}
        ≤ \sum_{s = k}^{δ|V|}\sum_{t = \frac{2s}{k}}^{|E|} \bigg(C\frac{n}{s}\bigg)^{s} \bigg(C'\frac{s^{k-1}}{n^{k-1}}\bigg)^{t}\\
        &≤ 2\sum_{s = k}^{δ|V|}\bigg(C\frac{n}{s}\bigg)^{s} \bigg(C'\frac{s^{k-1}}{n^{k-1}}\bigg)^{\frac{2s}{k}} = 2\sum_{s = k}^{δ|V|}\bigg(C''\frac{n^k}{s^k}\frac{s^{2k-2}}{n^{2k-2}}\bigg)^{\frac{s}{k}} = 2\sum_{s = k}^{δ|V|}\bigg(C'''\frac{s}{n}\bigg)^{\frac{s(k-2)}{k}}.
    \end{align*}
    To get rid of the summation over $t$, we assumed $(s/n)^{k-1} ≤ δ^{k-1} ≤ \frac{1}{2C'}$. Elementary arguments show that in the resulting bound, the contribution of summands for $s ∈ \{k,…,2k\}$ is of order $\O(\frac{1}{n})$, the contribution of the summands with $s ∈ \{2k+1,…,O(\log n)\}$ are of order $\O(\frac{\log n}{n²})$ (using $\frac{s}{n} ≤ \smash{\frac{\log n}{n}}$) and the contribution of the remaining terms with $s ≥ 3\log₂ n$ is of order $\O(2^{-\log₂ n}) = \O(\frac{1}{n})$  (using $C''' \frac{s}{n} ≤ C'''δ(ℓ+2) ≤ \frac{1}{2}$).
    
    This gives $\Pr[\bigcup_{s,t} W_{s,t}] = \O(n^{-1})$, proving the claim.
\end{proof}

We are ready to prove the “$\ero_k ≤ f_k$” of \cref{thm:main}, stated here as a theorem of its own.
\begin{theorem}
    \label{thm:erosion-implies-peeling}
    For all $k ≥ 3$ we have $\ero_k ≤ f_k$.
\end{theorem}

\begin{proof}
        We need to prove that for any $c < \ero_k$ and any $ℓ ∈ ℕ$ the fuse graph $F = F(n,k,c,ℓ)$ is peelable whp.

        First, let $δ = δ(k,ℓ)$ be the constant from Lemma \ref{lem:no-tiny-core} and $R = R(δ/2,ℓ)$ as well as $N = N(δ/2,ℓ)$ the corresponding constants from \cref{lem:back-to-peeling}(\ref{lemitem:almost-peelable}).
        
        Assuming $n ≥ N$ we have $\q{R}(i) ≤ δ/2$ for all $i ∈ I$, meaning any vertex $v$ from $F$ is \emph{not} deleted within $R$ rounds of $\peel_v(F)$ with probability at most $δ/2$. Since $\peel(F)$ deletes at least the vertices that any $\peel_v(F)$ for $v ∈ V$ deletes, the expected number of vertices not deleted by $\peel(F)$ within $R$ rounds is at most $δ|V|/2$.
        
        Now standard arguments using Azuma's inequality (see e.g. \cite[Theorem 13.7]{MU:Probability:2017}) suffice to conclude that whp at most $δ|V|$ vertices are not deleted by $\peel(F)$ within $R$ rounds.
        
        By Lemma \ref{lem:no-tiny-core} whp neither the remaining $δ|V|$ vertices, nor any of its subsets induces a hypergraph of minimum degree $2$. Therefore the core of $F$ is empty.
\end{proof}

A natural follow-up question to \cref{thm:erosion-implies-peeling} would be whether $\ero_k = f_k$, which would also imply $f_k ≤ \co_k$. To establish this stronger claim, we would have to exclude the possibility that for certain densities $c$ there is function $r(n) = ω(1)$ such that a constant fraction of vertices survive $r(n)$ rounds but are nevertheless deleted eventually. It seems plausible that arguments similar to \cite[Lemma 4]{Molloy05:Cores-in-random-hypergraphs} can be used, but since our main goal is reached we do not pursue this now.


\section{Approximating the Erosion and Consolidation Thresholds}
\label{sec:approximating}

We now approximate the thresholds $\ero_k$ (and analogously $\co_k$) with numerical methods. Note that if $c < \ero_k$ (if $c > \co_k$), then this can be verified in a finite computation, because the correct value of $R$, together with a bound on the required precision of floating point operations (when rounding conservatively), constitutes a witness. Moreover, the function $\hP^{r} \step_{ξ₁/2}^{1}$ can be represented by a finite number of reals, since it is constant on $(-∞,-(k-1)r]$ and constant on $[(k-1)r,∞)$.

To approximate $\ero_k$ (and $\co_k$) with high precision, more efficient approaches are required, however. We compute upper bounds on $\hP^r \step_{ξ₁/2}^{1}$ by focusing on a finite domain $[-D,D]$ for some $D ∈ ℕ$ and rounding conservatively outside of it. Concretely we define $(a_r : ℤ → [0,1])_{r ∈ ℕ₀}$ (dependent on $k$, $c$ and $D$) with $a₀ := \step_{ξ₁/2}^{1}$ (analogously $(b_r : ℤ → [0,1])_{r ∈ ℕ₀}$ with $b₀ := \step_{0}^{(ξ₁+ξ₂)/2}$). For $r ≥ 0$ we let
\[ a_{r+1}(i) := \begin{cases}
    a_{r+1}(-D) & \text{ if $i < -D$},\\
    \hP a_r(i) & \text{ if $-D ≤ i ≤ D$},\\
    1 & \text{ if $i > D$}.
\end{cases}
\qquad b_{r+1}(i) := \begin{cases}
    0 & \text{ if $i < -D$},\\
    \hP b_r(i) & \text{ if $-D ≤ i ≤ D$},\\
    \hP b_r(D) & \text{ if $i > D$}.
\end{cases}\]
Due to the limited effective domain, each $a_r$ is given by $2D+2$ values. It is easy to see that each $a_r$ is monotonous and fulfils $a_{r+1} ≤ \hP a_r$, which implies $\hP^r \step_{ξ₁/2}^{1} ≤ a_r$. If we find $a_r(0) < ξ₁/2$, then by monotonicity we have $a_r ≤ \sR \step_{ξ₁/2}^{1}$ and therefore:
\[ ∃R ∈ ℕ: a_R(0) < ξ₁/2 \quad ⇒ \quad ∃R ∈ ℕ:  \hP^R \step_{ξ₁/2}^1 < \sR \step_{ξ₁/2}^1 \quad \stackrel{\text{def}}{⇒} \quad c < \ero_k.\]
(Analogously if $b_R(-1) > (ξ₁+ξ₂)/2$ then $c > \co_k$ follows.)

\paragraph{Experimental Results.} For $D = 50$ and all $k ∈ \{3,…7\}$ we computed, using double-precision floating point values, $a₁,a₂,…$ and $b₁,b₂,…$ for various $c$. For each pair $(k,c)$, we either find that $\hP(k,c)$ is consolidating, it is eroding, or none of the two can be verified. The results suggest that $\ero_k < c_k^* < \co_k$ where $c_k^*$ is the orientability threshold for $k$-ary Erdős-Renyi hypergraphs.

Concretely, we considered for $j = 1,2,3,…$ the values $c_k^* - 2^{-j}$ and tried to verify that they are less than $\ero_k$. The largest for which we succeeded is report as $b_k$ in \cref{tab:upper-and-lower-bounds} on page \pageref{tab:upper-and-lower-bounds}.
The largest number of iterations required was $6·10^{7}$. For the first value that could not be shown to be less than $\ero_k$, our approximations of the sequence of $(a_i)_{i∈ℕ}$ became stationary with $a[0] > ξ₁/2$, i.e.~the double-precision floats did not change any more (the highest number of iterations to reach this point was $2·10^8$). It is possible that the value is still less than $\ero_k$ and our choice of $D$ or the precision of our floats is simply insufficient. Further experiments with $128$-bit floats and larger values of $D$ suggest however, that there is a tiny but real gap between $\ero_k, c_k^*$ and $\co_k$ and the natural conjecture of equality is misplaced.

In the same way we report the smallest value of the form $c_k^* + 2^{-j}$ for which we verified that it exceeds $\co_k$ as $B_k$ in \cref{tab:upper-and-lower-bounds}.


\section{Peeling Necessitates Orientability of Erd\texorpdfstring{ő}{ö}s-Renyi Hypergraphs}
\label{sec:upperbound}

We now prove the “$f_k ≤ c_k^*$”-half of \cref{thm:main}, stated as \cref{thm:non-orientability-implies-non-peelability}. Recall that an \emph{orientation} of a hypergraph $H = (V,E)$ is an injective map $f: E → V$ with $f(e) ∈ e$ for all $e ∈ E$ and $c_k^*$ is the thresholds for orientability $k$-uniform Erdős-Renyi hypergraphs.

After classical ($2$-ary) cuckoo hashing was discovered \cite{PR:Cuckoo:2004} (relying on $c₂^* = \frac 12$), the thresholds for $k > 2$ were determined independently by \cite{DGMMPR:Tight:2010,FP:Sharp:2012,FM:Maximum:2012}, with generalisations to other graphs and hypergraphs studied in \cite{CSW:The_Random:2007,FR:The_k-orientability:2007,Leconte:Cuckoo:2013,LP:3.5-Way:2009,Walzer:OverlappingBlocks:2018}.

Note that if $H$ is peelable then it is also orientable: Just orient each edge $e$ to a vertex $v ∈ e$ such that $v$ and $e$ are deleted in the same round of $\peel(H)$.

Our proof of \cref{thm:non-orientability-implies-non-peelability} relies strongly on a deep and remarkable theorem due to Lelarge \cite{L:A_New_Approach:2012}. To clarify its role in our proof, we restate it in weaker but sufficient form.
\begin{theorem}[{Lelarge \cite[Theorem 4.1]{L:A_New_Approach:2012}}]
    Let $(G_n = (A_n, B_n, E_n))_{n∈ℕ}$ be a sequence of bipartite graphs with $|E_n| = O(|A_n|)$. Let further $M(G_n)$ be the size of a maximum matching in $G_n$. If the \emph{random weak limit} $ρ$ of $(G_n)_{n∈ℕ}$ is a \emph{bipartite unimodular Galton-Watson tree}, then $\lim\limits_{n→∞} \frac{M(G_n)}{|A_n|}$ exists almost surely and depends only on $ρ$.
\end{theorem}
To see the connection, note that an orientation of a hypergraph is a left-perfect matching in its (bipartite) incidence graph.

\begin{theorem}
    \label{thm:non-orientability-implies-non-peelability}
    For all $k ≥ 3$ we have $f_k ≤ c_k^*$.
\end{theorem}

\begin{proof}
    Let $c = c_k^* + ε$. We need to show that there exists $ℓ ∈ ℕ$ such that the fuse graph $F=F(n,k,c,ℓ)$ is not peelable whp.
    
    Let $H = H^k_{n,cn}$ be the $k$-ary Erdős-Renyi random hypergraph with density $c$. By choice of $c$, $H$ is not orientable whp. More strongly even, there exists $δ = δ(ε) > 0$ such that the largest \emph{partial orientation}, i.e.~the largest subset of the edges that can be oriented, has size $(1-δ)cn + o(n)$ whp, see for instance \cite{L:A_New_Approach:2012}.
    
    \def\TF{\widetilde{F}}
    We set $ℓ = \frac{k}{δc}$ and consider $F$ as well as the hypergraph $\TF$ where the vertices $i$ and $i+nℓ$ for all $i ∈ \{1,…,(k-1)n\}$ are merged. This “glues” the last $k-1$ segments of $F$ on top of the first $k-1$ segments of $F$, making $\TF$ a “seamless” version of our construction. Crucially, the \emph{random weak limit} of $\TF$ and $H$ coincide, i.e.~for any constant $R ∈ ℕ$ the distribution of the $R$-neighbourhood $N_{\TF}(v,R)$ of a random vertex $v$ of $\TF$ has the same limit (as $n → ∞$) as the distribution of the $R$-neighbourhood $N_H(v,R)$ of a random vertex $v$ of $H$.\footnote{The common limit of the incidence graphs of $\TF$ and $H$ is the bipartite unimodular Galton-Watson tree described in \cite[Section 4]{L:A_New_Approach:2012}. Standard arguments, e.g. from \cite{K:Poisson:2006,Leconte:Cuckoo:2013} suffice to establish the identity.}
    It now follows from \cite[Theorem 4.1]{L:A_New_Approach:2012} that the size of the largest partial orientation of $\TF$  is essentially also a $(1-δ)$-fraction of the number of edges, namely $(1-δ)cℓn+o(n)$ whp. Switching from $\TF$ back to $F$ can increase the size of a largest partial orientation by at most $(k-1)n$ to $(1-δ+\frac{k-1}{cℓ})cℓn+o(n) = (1-\frac{δ}{k})cℓn+o(n)$ whp. Thus $F$ is not orientable whp and therefore not peelable whp.
\end{proof}

\section{Experiments}
\label{sec:experiments}
\def\construct{\textsf{construct}\xspace}
\def\query{\textsf{query}\xspace}

We used our hypergraphs to implement retrieval data structures and compare it to existing implementations.

A \emph{$1$-bit retrieval data structure} for a universe $\U$ is a pair of algorithms \construct and \query, where the input of \construct is a set $S ⊆ \U$ of size $m = |S|$ and $f: S → \{0,1\}$. If \construct succeeds, then the output is a data structure $D_f$ such that $\query(D_f,x) = f(x)$ for all $x ∈ S$. The output of $\query(D_f,y)$ for $y ∈ \U \setminus S$ may yield an arbitrary element of $\{0,1\}$. The interesting setting is when the data structure may only occupy $\O(m)$ bits. See \cite{BPZ:Practical:2013,BPZ:Simple:2007,DP:Succinct:2008,Vigna:Fast-Scalable-Construction-of-Functions:2016,P:An_Optimal:2009}.

One approach is to map each element $x ∈ S$ to a set $e_x ⊂ [N]$ via a hash function, where $N = m/c$ for some desired edge density $c$. One then seeks a solution $z : [N] → \{0,1\}$ satisfying $\bigoplus_{v ∈ e_x} z(v) = f(x)$ for all $x ∈ S$. The bit-vector $z$ and the hash function then form $D_f$. A query simply evaluates the left hand side of the equation for $x$ to recover $f(x)$. To compute $z$, we consider the hypergraph $H = ([N], \{e_x, x ∈ S\})$. A peelable vertex $v ∈ [N]$ only contained in one edge $e_x$ corresponds to a variable $z(v)$ only occuring in the equation associated with~$x$. It is thus easy to see that if $H$ is peelable, repeated elimination and back-substitution yields $z$ in $\O(m)$ time.

We implemented the following peeling-based variations and report results in \footnotemark\ in \cref{tab:results}. By the \emph{overhead} of an implementation we mean $\frac{N'}{m}-1$ where $N' ≥ N$ is the total number of bits used, including auxiliary data structures.
\footnotetext{%
    Experiments were performed on an desktop computer with an Intel${}^{\textrm{\textregistered}}$\,Core i7-2600 Processor @~3.40GHz. In all cases, the data set $S$ contains the first $m = 10^7$ URLs from the \texttt{eu-2015-host} dataset gathered by \cite{BMSV:Crawls:2014} with ${≈}$80 bytes per key, and $f \colon \U → \{0,1\}$ is taken to be the parity of the string length. As hash function we used \textsc{MurmurHash3\_x64\_128} \cite{Appleby:MurmurHash3:2012}. If more than 128 hash bits were needed, techniques resembling double-hashing were used to generate additional bits to avoid another execution of murmur. Reported query times are averages obtained by querying all elements of the data set once. They include the roughly 25 ns needed to evaluate murmur on average. The reported numbers are medians of 5 executions. 
}
\begin{table}
    \begin{tabular}{ccrrr}
        \toprule
        & \multirow{2}{*}{Configuration} & \multirow{2}{*}{Overhead} & \construct                              & \query\\
        &                        &                 &  {\small $[\ms/\textrm{key}]$}&{\small $[\ns]$}\\\midrule
        Botelho et al. \cite{BPZ:Practical:2013} & c = 0.81 & 23.5\% & 0.32 & 59\\
        $\langle$Fuse Graphs$\rangle$ &$c = 0.910, k = 3, ℓ = 100$&12.1\%&0.29&55\\
        $\langle$Fuse Graphs$\rangle$ &$c = 0.960, k = 4, ℓ = 200$&5.7\%&0.29&60\\
        $\langle$Fuse Graphs$\rangle$ &$c = 0.985, k = 7, ℓ = 500$&2.7\%&0.38&74\\
        \hdashline
        Luby et al. \cite{LMSS:Efficient_Erasure:2001} & $c = 0.9, D = 12$ & 11.1\% & 0.79 & 94\\
        Luby et al. \cite{LMSS:Efficient_Erasure:2001} & $c = 0.99, D = 150$ & 1.1\% & 0.87 & 109\\
        \midrule
        Genuzio et al. \cite{Vigna:Fast-Scalable-Construction-of-Functions:2016} &$c = 0.91, k = 3, C = 10^4$&10.2\%&1.30&58\\
        Genuzio et al. \cite{Vigna:Fast-Scalable-Construction-of-Functions:2016} &$c = 0.97, k = 4, C= 10^4$&3.4\%&2.20&64\\
        the authors \cite{DW:Retrieval-log-extra-bits:2019} &$c = 0.9995, ℓ = 16, C=10^4$&0.25\%&2.47&56\\
        \bottomrule
    \end{tabular}
    \caption{Overheads and average running times per key of various practical retrieval data structures.}
    \label{tab:results}
\end{table}
\begin{description}
        •[Botelho et al. \cite{BPZ:Practical:2013}] $H$ is a $3$-ary Erdős-Renyi hypergraph with an edge density below the peelability threshold $c₃ ≈ 0.818$. Construction via peeling and queries are very fast, but the overhead of $23\%$ is sizeable (i.e.~$D_f$ occupies roughly $1.23m$ bits).
        •[Fuse Graphs.] The edges are distributed such that $H$ is a fuse graph. Recall that the edge density is $c\frac{ℓ}{ℓ+k-1}$. Note that we let $ℓ$ grow with $k$ to keep the density close to $c$. We still keep $ℓ$ in a moderate range, as our construction relies on $n \gg ℓ$.
        •[Luby et al. \cite{LMSS:Efficient_Erasure:2001}] The edges are distributed such that $H$ is the peelable hypergraph from \cite{LMSS:Efficient_Erasure:2001} already mentioned on page \pageref{page:LMSS}. To our knowledge these hypergraphs have not been considered in the context of retrieval. They seem to be particularly well suited to achieve very small overheads at the cost of larger construction and mean query times compared to our other approaches. Note that the largest edge size is $D+4$ and the worst-case query time therefore much larger than the reported average query time.
\end{description}
For reference, we also implemented two recent retrieval data structures that do not rely on peeling but solve linear systems \cite{DW:Retrieval-log-extra-bits:2019,Vigna:Fast-Scalable-Construction-of-Functions:2016}. There, to counteract cubic solving time, the input is partitioned into chunks of size $C$. Especially \cite{DW:Retrieval-log-extra-bits:2019} achieves much smaller overheads than what is feasible with peeling approaches, with the downside of being much slower and more complicated.

Overall, it seems using fuse graphs in retrieval data structures has a chance of outperforming existing approaches when moderate memory overheads of $≈ 5\%$ are acceptable.

However, more research is required to explore the complex space of possible input sizes, configurations of the data structures and trade-offs between overhead and runtime. Our implementations are configured reasonably, but arbitrary in some aspects. A full discussion is beyond the scope of this paper.





\section{Conclusion}
\label{sec:conclusion}

We introduced for all $k ∈ ℕ$ a new family of $k$-uniform hypergraphs where the vertex set is partitioned into a large but constant number of segments. Each edge chooses a random range of $k$ consecutive segments and one random incidence in each of them.

\def\conv{\ \smash{\scalebox{0.8}{$\stackrel{k → ∞}{\longrightarrow}$}}\ }

While we have no asymptotic results on the resulting peelability thresholds $f_k$, at least for small $k$ they are remarkably close to $c_k^*$ with $0 ≤ c_k^* - f_k ≤ 10^{-5}$ for $k ∈ \{3,4,5,6,7\}$. In other words, $f_k$ almost coincides with the \emph{orientability} threshold $c_k^*$ of Erdős-Renyi hypergraphs and significantly exceeds their peelability threshold $c_k$. Note that $c_k^* = 1 - (1+o_k(1))e^{-k} \conv 1$ (see \cite[page 3]{FP:Sharp:2012}) while $c_k \conv 0$ (see e.g. \cite{Molloy05:Cores-in-random-hypergraphs}). When plugging our hypergraphs into the retrieval framework by~\cite{BPZ:Practical:2013}, we obtained corresponding improvements with respect to memory usage, with no discernible downsides.

\subparagraph{Future Experiments.} While our experiments on retrieval data structures are  promising, it is unclear how robustly the advantages translate to other practical settings where peelable hypergraphs are used, say when implementing Invertible Bloom Lookup Tables \cite{GM:Invertable:2011}. There are hidden disadvantages of our hypergraphs not considered in this paper – for instance the number of rounds needed to peel our hypergraphs is higher, possibly hurting parallel peeling algorithms – as well as hidden advantages – peeling in external memory, a setting considered in \cite{BBOVV:Cache-Oblivious-Peeling:14}, is easy due to the locality of the edges.

\subparagraph{A Theoretical Question.} Given our results, it is natural to suspect a fundamental connection between $f_k$ and $c_k^*$. Quite possibly, the tiny gap that seems to remain between the values – clearly negligible from a practical perspective – is merely an artefact of the discreteness of segments in our construction.

This discreteness, while heavily used in our arguments, may in fact be dispensable. Indeed, we believe the key idea behind our hypergraphs is \emph{limited bandwidth} where a hypergraph on vertex set $[n]$ has bandwidth at most $d$ if each edge $e$ satisfies $\max_{v ∈ e} v - \min_{v ∈ e} v < d$ (the incidence matrix can then be sorted to resemble a bandmatrix). Such a hypergraph can be generated by choosing for each edge a random range of $d$ consecutive vertices and $k$ incidences independently and uniformly at random from that range. In experiments with $k = 3$ and $d = εn$, such hypergraphs performed similar to the hypergraphs we analysed (with $k = 3$ and $ℓ≈1/ε$). Note that there are no discrete segments in the modified construction. It would be nice to see whether in such a variation peelability and orientability are more elegantly and more intimitely linked.

    \bibliography{bibliographie}
		
\end{document}